\documentclass[11pt]{article}
\usepackage{fullpage}
\usepackage{latexsym}
\usepackage{epsf}

\newtheorem{theorem}{Theorem}[section]
\newtheorem{lemma}{Lemma}[section]

\newcommand{\qed}{\hspace*{\fill}$\Box$\par\medskip}
\newenvironment{proof}{\noindent{\it Proof. }\ignorespaces}{\qed}

\title{{\bf Improved Randomized Online Scheduling of Intervals and Jobs}%
\thanks{The work described in this paper was fully supported by
grants from the Research Grant Council of the Hong Kong SAR,
China [CityU 119307] and NSFC Grant No. 60736027 and 70702030.}
}

\author{
Stanley P. Y. Fung%
\footnote{
Department of Computer Science, University of Leicester,
Leicester, United Kingdom. Email: {\tt pyfung@mcs.le.ac.uk}}
\and
Chung Keung Poon%
\footnote{
Department of Computer Science, City University of Hong Kong, Hong Kong.
Email: {\tt csckpoon@cityu.edu.hk}}
\and
Feifeng Zheng%
\footnote{
School of Management, Xi'an Jiaotong University, Xi'an, China. 
Email: {\tt zhengff@mail.xjtu.edu.cn}}
}

\begin{document}

\maketitle

\begin{abstract}
We study the online preemptive scheduling of intervals and jobs (with
restarts).  Each interval or job has an arrival time, a deadline,
a length and a weight.  The objective is to maximize the total weight of
completed intervals or jobs.  While the deterministic case for intervals
was settled a long time ago, the randomized case remains open.  In this
paper we first give a 2-competitive randomized algorithm for the case of 
equal length intervals.  The algorithm is barely random in the sense that 
it randomly chooses between two deterministic algorithms at the 
beginning and then sticks with it thereafter.  
Then we extend the algorithm to cover several other cases of interval 
scheduling including monotone instances, C-benevolent instances and 
D-benevolent instances, giving the same competitive ratio.
These algorithms are surprisingly simple but have the best competitive 
ratio against all previous (fully or barely) randomized algorithms.  
Next we extend the idea to give
a 3-competitive algorithm for equal length jobs. 
Finally, we prove a lower bound of 2 on the competitive ratio of all 
barely random algorithms that choose between two deterministic algorithms 
for scheduling equal length intervals (and hence jobs). \\

\noindent
{\bf keywords}: 
interval and job scheduling; preemption with restart;
online algorithms; randomized; lower bound

\end{abstract}

\section{Introduction}

In this paper, we study two online preemptive scheduling problems.
In the {\em interval scheduling problem},
we are to schedule a set of weighted intervals which arrive online
(in the order of their left endpoints)
so that at any moment, at most one interval is being processed.
We can abort the interval currently being processed in order to
start a new one.
The goal is to maximize the sum of the weights of completed
intervals.
The problem can be viewed as a job scheduling problem in which
each job has, besides its weight, an arrival time, a length and a deadline.
Moreover, the deadline is always tight,
i.e., deadline always equals arrival time plus length.
Thus, if one does not start an interval immediately upon its
arrival, or if one aborts it before its completion,
that interval will never be completed.
The problem is fundamental in scheduling and
is clearly relevant to a number of online problems
such as call control and bandwidth allocation
(see e.g., \cite{ABFR94,CanIra98,Woeg94}).

We also study the more general problem of
{\em job scheduling with restart}.
Here, the deadline of a job needs not be tight and
we can abort a job and restart it from the beginning some time later.
Both problems are in fact special cases of the {\em broadcast
scheduling problem} which gains much attention recently
due to its application in video-on-demand, stock market quotation, etc
(see e.g., \cite{KimChw04,Ting06,ZFCCPW06}).
In that problem, a server holding a number of pages
receives requests from its clients and schedules the
broadcasting of its pages.
A request is satisfied if the requested page is
broadcasted in its entirety after the arrival time and 
before the deadline of the request.
The page currently being broadcasted can be aborted
in order to start a new one,
and the aborted page can be re-broadcasted from the beginning later.
Interval and job scheduling with restart can be seen
as a special case in which each request asks for a different page.

Our results concern {\em barely random} algorithms,
 i.e., randomized algorithms that randomly choose from a very small
(constant) number of deterministic algorithms at the beginning
and then stick with it thereafter.
Quite some previous work in online scheduling considered the use of
barely random algorithms (see e.g. 
\cite{Albers02,CJST07,Sei98});
it is interesting to consider how the competitiveness improves
(upon their deterministic counterparts) 
by combining just a few deterministic algorithms.
From now on, whenever we refer to ``barely random algorithms'',
we mean algorithms that choose between {\em two} deterministic algorithms
but possibly with unequal probability.

\paragraph{Types of instances.}

In this paper, we consider the following special types of intervals or jobs:
\begin{enumerate}
\item
equal length instances where all intervals or jobs have the same
length, 
\item
monotone instances where intervals arriving earlier also have
earlier deadlines, and 
\item
C- and D-benevolent instances where 
the weight of an interval is given by some `nice' function of its length
(convex increasing for C-benevolent, and decreasing for D-benevolent).
\end{enumerate}
The models will be defined precisely in the next section.
These cases are already highly non-trivial, 
as we will see shortly, and many previous works on these
problems put further restrictions on the inputs
(such as requiring jobs to be unweighted
or arrival times to be integral, in addition to being equal-length).
The power of randomization for these problems is
especially unclear.

\subsection{Previous work}

The general case where intervals can have arbitrary lengths and weights
does not admit constant competitive algorithms \cite{Woeg94}, even with 
randomization \cite{CanIra98}.  Therefore, some special types of 
instances have been studied in the literature.

We first mention results for equal length interval scheduling.
The deterministic case was settled in \cite{Woeg94} where a 
4-competitive algorithm and a matching lower bound were given.
Miyazawa and Erlebach \cite{MiyErl04} were the first to give a better
randomized algorithm: its competitive ratio is 3
but it only works for a special case where the weights of the intervals 
form a non-decreasing sequence.
They also gave the first randomized lower bound of 5/4.
The first randomized algorithm for arbitrary weight that
has competitive ratio better than 4 (the bound for deterministic
algorithms) was devised in \cite{FuPoZh07}.
It is 3.618-competitive and is barely random,
choosing between two deterministic algorithms with equal probability.
In the same paper, a lower bound of 2 for such barely random algorithms
and a lower bound of 4/3 for general randomized algorithms
were also proved.
Recently, Epstein and Levin \cite{EpsLev10} 
gave a 2.455-competitive randomized algorithm
and a 3.227-competitive barely random algorithm.
They also gave a 1.693 lower bound on the randomized competitive ratio.

The class of monotone instances (also called {\it similarly ordered}
\cite{CJST07} or {\it agreeable} \cite{LiSS05} instances in the literature)
is a generalization of the class of equal length instances.
Therefore, the former class inherits all the lower bounds for the latter
class.
In the offline case, the class of monotone instances is actually
equivalent to that of equal length instances
because of the result (see e.g. \cite{BogWes99})
that the class of proper interval graphs 
(intersection graphs of intervals where no interval is
strictly contained in another)
is equal to the class of unit interval graphs.
In the online case however, it is not completely clear that 
such an equivalence holds although some of the algorithms 
for the equal length case also work for the monotone case
(e.g. \cite{MiyErl04,FuPoZh07,EpsLev10}).

Some of the aforementioned results for equal length instances also
work for C- and D-benevolent instances, including
Woeginger's 4-competitive deterministic algorithm,
the lower bound of 4/3 in \cite{FuPoZh07}\footnote{
This and most other lower bounds for D-benevolent instances
only work for a subclass of functions that satisfy a surjective
condition.},
the upper bounds in \cite{EpsLev10} (for D-benevolent instances only)
and the lower bound in \cite{EpsLev10} (for C-benevolent instances only;
they gave another slightly weaker lower bound of 3/2 for D-benevolent
instances).
A 3.732-competitive barely random algorithm
for C-benevolent instances was given by Seiden \cite{Sei98}.  
Table \ref{tab:ResultsInt} summarizes the various upper and lower bounds
for randomized interval scheduling.

\begin{table}
\centerline{
\begin{tabular}{|c|c|c|} \hline
  & upper bound                           & lower bound              \\ \hline
equal length  
  & 2.455 \cite{EpsLev10}                 & 1.693 \cite{EpsLev10}    \\
  & 3.227 (barely random) \cite{EpsLev10} &                          \\
  & 2 (barely random) [this paper]        & 2 (barely random) [this paper] \\ \hline
monotone     
  & same as above                         & same as above            \\ \hline
C-benevolent 
  & 3.732 \cite{Sei98}                    & 1.693 \cite{EpsLev10}    \\
  & 2 (barely random) [this paper]        &                          \\ \hline
D-benevolent 
  & 2.455 \cite{EpsLev10}                 & 1.5 \cite{EpsLev10}
                                            (with a surjective condition) \\
  & 3.227 (barely random) \cite{EpsLev10} &                          \\ 
  & 2 (barely random) [this paper]        &                          \\ \hline
\end{tabular}
}
\caption{Best previous and new results for randomized interval scheduling}
\label{tab:ResultsInt}
\end{table}

Next we consider the problem of job scheduling with restarts. 
Zheng et al. \cite{ZFCCPW06} gave a 4.56-competitive deterministic algorithm. 
The algorithm was for the more general problem
of scheduling broadcasts but it works for jobs scheduling
with restarts too.
We are not aware of previous results in the randomized case.
Nevertheless, Chrobak et al. \cite{CJST07} 
considered a special case where the jobs have no weights 
and the objective is to maximize the number of completed jobs.  
For the randomized nonpreemptive case 
they gave a 5/3-competitive barely random algorithm 
and a lower bound of 3/2 for barely random algorithms that choose between
two deterministic algorithms.
They also gave an optimal 3/2-competitive
algorithm for the deterministic preemptive (with restart) case, and
a lower bound of 6/5 for the randomized preemptive case.

We can also assume that the time is discretized into unit length slots
and all (unit) jobs can only start at the beginning of each slot.
Being a special case of the problem we consider in this paper,
this version of unit job scheduling has been widely studied
and has applications in buffer management of QoS switches.
For this problem, a $e/(e-1)$-competitive randomized algorithm was given
in \cite{CCFJST06},
and a randomized lower bound of 1.25 was given in \cite{ChiFun03}.  
The current best deterministic algorithm is 1.828-competitive \cite{EngWes07}.

An alternative preemption model is to allow the partially-executed job to
resume its execution from the point that it is preempted.  This was
studied, for example, in \cite{BKMMR+92,KorSha95}.

\subsection{Our results}

In this paper we give new randomized algorithms for the different
versions of the online interval scheduling problem.
They are all barely random and have a competitive ratio of 2.
Thus they substantially improve previous results. 
See Table \ref{tab:ResultsInt}.
It should be noted that although the algorithms are fairly
simple, they were not discovered in 
several previous attempts by other researchers and ourselves
\cite{EpsLev10,FuPoZh07,MiyErl04}.
Moreover the algorithms for all these versions of the problem are based on 
the same idea, which gives a unified way of analyzing these algorithms 
that were not present in previous works.

Next we extend the algorithm to the case of job scheduling (with restarts),
and prove that it is 3-competitive.
This is the first randomized algorithm we are aware of for this problem.
The extension of the algorithm is very natural but the proof is
considerably more involved.

Finally we prove a lower bound of 2 for barely random algorithms for 
scheduling equal length intervals (and jobs)
that choose between two deterministic algorithms, not necessarily with
equal probability.  Thus it matches the upper bound of 2 for this
class of barely random algorithms.
Although this lower bound does not cover more general classes of
barely random or randomized algorithms, we believe that this is still
of interest.
For example, a result of this type appeared in \cite{CJST07}.
Also, no barely random algorithm using three or more deterministic
algorithms with a better performance is known.
The proof is also much more complicated than the
one in \cite{FuPoZh07} with equal probability assumption.

\section{Preliminaries}

A {\it job} $J$ is specified by its arrival time $r(J)$, its deadline $d(J)$, 
its length (or processing time) $p(J)$ and its weight $w(J)$.  
All $r(J), d(J), p(J)$ and $w(J)$ are nonnegative real numbers.
An {\it interval} is a job with tight deadline, i.e. $d(J) = r(J) + p(J)$.
We further introduce the following concepts for intervals: for intervals $I$
and $J$ with $r(I) < r(J)$, $I$ \textit{contains} $J$ if
$d(I) \geq d(J)$; if $r(J) < d (I) < d(J)$, the two
intervals \textit{overlap}; and if $d(I) \le r(J) < d(J)$, the intervals
are {\it disjoint}.

Next we define the types of instances that we consider in this paper.
The {\it equal length} case is where $p(J)$ is the same for all $J$;
without loss of generality we can assume $p(J)=1$.
The remaining notions apply to intervals only.
An instance is called {\it monotone} if for any two intervals $I$ and
$J$, if $r(I) < r(J)$ then $d(I) \leq d(J)$.  
An instance is called {\it C-benevolent} if the weights of intervals are
given by a function $f$ of their lengths, where the function $f$ satisfies
the following three properties: 
\begin{description}
\item[(i)]
$f(0)=0$ and $f(p)>0$ for all $p>0$, 
\item[(ii)]
$f$ is strictly increasing, and 
\item[(iii)]
$f$ is convex,
i.e.  $f(p_1)+f(p_2) \le f(p_1-\epsilon)+f(p_2+\epsilon)$ for 
$0 < \epsilon \le p_1 \le p_2$.
\end{description}
Finally, an instance is called {\it D-benevolent} if the weights of 
intervals are given by a function $f$ of their lengths where
\begin{description}
\item[(i)]
$f(0)=0$ and $f(p)>0$ for any $p>0$, and
\item[(ii)]
$f$ is decreasing in $(0, \infty)$.
\end{description}

In our analysis, we partition the time axis into segments called
{\it slots}, $s_1, s_2, \ldots$,
such that each time instant belongs to exactly one slot
and the union of all slots cover the entire time axis.
The precise way of defining the slots depends on the case being studied 
(equal-length, monotone, C- or D-benevolent instances).
Slot $s_i$ is an {\it odd slot} if $i$ is odd, and is an
{\it even slot} otherwise.  

The following is an important, though perhaps unusual, definition 
used throughout the paper.
We say that a job (or an interval) is {\it accepted} by an algorithm $A$ 
in a slot $s$
if it is started by $A$ within the duration of slot $s$ and is then 
completed without interruption.  
Note that the completion time may well be after slot $s$.
$A$ may start more than one job in a slot,
but it will become clear that for all online algorithms that we consider,
at most one job will be accepted in a slot; all other jobs that were started
will be aborted.
For $OPT$ we can assume that it always completes each
interval or job it starts.

The {\it value} of a schedule is the total weight of the jobs that are 
completed in the schedule.
The performance of online algorithms is measured using
competitive analysis \cite{BorElY98}.  
An online randomized algorithm $A$ is {\it $c$-competitive}
if the expected value obtained by $A$ is at least $1/c$
the value obtained by the optimal offline algorithm, for any input instance.
The infimum of all such $c$ is called the {\it competitive ratio} of $A$.
We use $OPT$ to denote the optimal algorithm (and its schedule).

\section{Algorithms for Scheduling Intervals}

\subsection{Equal Length Instances}

In this section we describe and analyse a very simple algorithm 
$RAN$ for the case of equal length intervals.
$RAN$ is barely random and consists of two 
deterministic algorithms $A$ and $B$, described as follows.  
The time axis is divided into unit length slots, $s_1, s_2, \ldots$,
where slot $s_i$ covers time [$i-1,i$) for $i=1,2,\ldots$.  
Intuitively, $A$ takes care of odd slots and $B$ takes care of even slots.
Within each odd slot $s_i$, $A$ starts the interval arriving first.
If a new interval arrives in this slot while an interval 
is being processed, $A$ will abort and start the new interval 
if its weight is larger than the current interval;
otherwise the new interval is discarded.  
At the end of this slot,
$A$ is running (or about to complete) an interval with the largest weight 
among those that arrive within $s_i$; let $I_i$ denote this interval.
$A$ then runs $I_i$ to completion without abortion during the 
next (even) slot $s_{i+1}$.  
(Thus, $I_i$ is the only interval accepted by $A$ in slot $s_i$.)
Algorithm $A$ then stays idle until the beginning of the next odd slot.
$B$ runs similarly on even slots.
$RAN$ chooses one of $A$ and $B$ with equal probability 1/2 at the beginning.

\begin{theorem}
$RAN$ is $2$-competitive for online interval scheduling on equal length 
instances.
\end{theorem}
\begin{proof}
Each $I_i$ is accepted by either $A$ or $B$.  
Therefore, $RAN$ completes each $I_i$ with probability 1/2. 
On the other hand, $OPT$ can accept
at most one interval in each slot $s_i$,
with weight at most $w(I_i)$.
It follows that the total value of $OPT$ is at most 2 times the
expected value of $RAN$.
\end{proof}

Trivial examples can show that $RAN$ is not better than 2-competitive
(e.g. a single interval).  In fact we will show in Section~\ref{sec:LB} 
that no barely random 
algorithm that chooses between two deterministic algorithms is better
than 2-competitive.
But first we consider how this result can be generalized to other
types of instances.

\subsection{Monotone Instances}

\paragraph{Algorithm $RAN$-$M$.} 
We adapt the idea of $RAN$ to the case of 
monotone instances and call the algorithm $RAN$-$M$.
Similar to $RAN$, $RAN$-$M$ consists of two deterministic algorithms $A$ 
and $B$, each chosen to execute with probability 1/2 at the beginning.
The difference is that we cannot use the idea of unit length slots
but we must define the lengths of the slots in an online manner.

The execution of the algorithm is divided into {\em phases}
and we name the slots in each phase locally as $s_1, s_2, \ldots$
independent of other phases.
After the end of a phase and before the beginning of the next phase,
the algorithm (both $A$ and $B$) is idle with no pending intervals.
A new phase starts when the first interval arrives while the algorithm is 
idle.
Among all intervals that arrive at this time instant, 
let $I_0$ be the one with the earliest deadline (ties broken arbitrarily).
Then slot $s_1$ is defined as $[r(I_0), d(I_0))$.
$A$ aims to accept the heaviest interval
among those with arrival time falling within slot $s_1$.
To do this, $A$ simply starts the first interval arriving in $s_1$, and
then whenever a new interval arrives that is heavier than the
interval that $A$ is currently executing,
$A$ aborts the current one and starts the new heavier interval.
This is repeated until the time $d(I_0)$ is reached.
By the property of monotone instances and 
the choice of $I_0$,
these intervals all have finishing time on or after $d(I_0)$.  
Let $I_1$ denote the interval that $A$ is executing (or about to complete)
at the end of slot $s_1$, i.e., time $d(I_0)$.  
$B$ remains idle during the whole slot.
If $A$ just finishes $I_1$ at time $d(I_0)$, then it will become idle again
and this phase ends.  
Otherwise, $d(I_1) > d(I_0)$ and 
slot $s_2$ is now defined as $[d(I_0), d(I_1))$.

In slot $s_2$, $A$ continues to execute $I_1$ to completion without any 
interruption.
(Thus, $I_1$ is the only interval accepted by $A$ in slot $s_1$.)
$B$ accepts the heaviest interval among 
those with arrival time falling within slot $s_2$, in the same manner $A$
did in the previous slot.
This interval is denoted by $I_2$ and $B$ will
run it to completion during slot $s_3$
(if its deadline is after the end of slot $s_2$).

In general, slot $s_i$ (where $i>1$)
is defined as $[d(I_{i-2}), d(I_{i-1}))$.  
If $i$ is odd, then at the beginning of slot $s_i$,
$B$ is executing $I_{i-1}$ (the interval accepted by $B$ in slot $s_{i-1}$) 
and $A$ is idle.
$B$ will run $I_{i-1}$ to completion while
$A$ will accept the heaviest interval among those arriving
during this slot.
If $i$ is even, the actions are the same except that
the roles of $A$ and $B$ are reversed.

\begin{theorem}
RAN-M is $2$-competitive for online interval scheduling on monotone instances.
\end{theorem}
\begin{proof}
No interval will arrive during the idle time between phases (since otherwise
$RAN$-$M$ would have started a new phase), so each phase can be analyzed 
separately. Each interval completed by $OPT$ will be analyzed according
to the slot its arrival time falls into.

In each slot $s_i$, $OPT$ can accept at most one interval:
This is true for $s_1$ by the way $s_1$ is chosen.
For $i>1$,
consider the first interval $I'$ accepted by $OPT$ in slot
$s_i = [d(I_{i-2}), d(I_{i-1}))$.
(Recall that accepting a job means starting the job and then
executing it to completion without interruption.)
Since the start of slot $s_i$ is after $r(I_{i-1})$,
we have $r(I') > r(I_{i-1})$.
By the monotone property, $d(I') \geq d(I_{i-1})$.
So, $OPT$ cannot accept another interval in slot $s_i$.
The rest of the proof is the same as the equal length case,
namely, that the interval accepted 
by $OPT$ in each slot has weight at most that of the interval
accepted by $A$ or $B$ in the same slot.  
It follows that $RAN$-$M$ is 2-competitive.  
\end{proof}

\subsection{C-benevolent Instances}

\paragraph{Algorithm $RAN$-$C$.}
Once again, the algorithm for C-benevolent instances
$RAN$-$C$ consists of two deterministic algorithms $A$ and $B$,
each with probability 1/2 of being executed.
The execution of the algorithm is divided into phases as in the monotone
case.

When a new phase begins,
the earliest arriving interval, denoted by $I_0$, defines the
first slot $s_1$, i.e., $s_1 = [r(I_0), d(I_0))$.
(If there are several intervals arriving at the same time,
let $I_0$ be the one with the longest length.)
We first describe the processing of intervals in slot $s_1$,
which is slightly different from the other slots.
First, $B$ starts and completes $I_0$.  
During $s_1$, $A$ accepts the longest interval among those
with arrival time during $(r(I_0), d(I_0))$ and finishing time 
after $d(I_0)$. 
Denote this interval by $I_1$.
(Note that there may be other intervals that arrive and end 
before $I_1$ arrives. 
Naturally, $A$ could finish them in order to gain more value.
However, to simplify our analysis, we assume that $A$ will not
process them.)
If there is no such $I_1$,
i.e., no interval arrives within $s_1$ and ends after $d(I_0)$,
the phase ends at the end of $s_1$.

Suppose $I_1$ exists.
Then define slot $s_2$ as $[d(I_0), d(I_1))$.
$A$ uses the entire slot $s_2$ to complete $I_1$ without interruption.
After completing $I_0$ at time $d(I_0)$, 
$B$ accepts the longest interval (denoted $I_2$)
among those arriving within slot $s_2$ and finishing after $d(I_1)$,
in a way similar to the action of $A$ in the previous slot.
Again, if such an $I_2$ does not exist, 
the phase ends at the end of $s_2$.
Otherwise, slot $s_3$ is defined as $[d(I_1), d(I_2))$
and $B$ will complete $I_2$ that ends after $d(I_1)$.
Similarly, after $A$ finishes $I_1$ in time $d(I_1)$,
it starts the longest interval (denoted by $I_3$)
arriving during $s_3$ and finishing after $d(I_2)$, and so on.

In general, slot $s_i$ (for $i>1$) is defined as $[d(I_{i-2}), d(I_{i-1}))$.
If $i$ is odd, then $B$ takes the entire slot to complete the interval
$I_{i-1}$ without interruption while $A$ accepts the longest interval 
$I_{i}$ that arrives during slot $s_i$ and ends after $d(I_{i-1})$.  
If $i$ is even then the roles of $A$ and $B$ are reversed.

\paragraph{Competitive Analysis.}

We first state the
following useful lemma which holds for any C-benevolent function $f$.

\begin{lemma} \label{l1}
For any C-benevolent function $f$, given any $k+1$ positive
real numbers $p_i$ $(1\leq i\leq k)$ and $P$, if $P\geq \sum^k_{i=1}
p_i$, then $f(P)\geq \sum^k_{i=1} f(p_i)$.
\end{lemma}
\begin{proof}
$f(P)\geq f(\sum^k_{i=1} p_i) \geq f(p_1)+ f(\sum^k_{i=2}
p_i)\geq \sum^2_{i=1}f(p_i)+f(\sum^k_{i=3} p_i)\geq \ldots\geq
\sum^k_{i=1} f(p_i)$.  
\end{proof}

\begin{theorem} 
$RAN$-$C$ is $2$-competitive for online interval scheduling on 
C-benevolent instances.
\end{theorem}
\begin{proof}
As a first step to the proof we simplify the $OPT$ schedule.
Within each slot $s_i$ in a phase, $i \ge 1$,
$OPT$ starts a sequence of disjoint intervals 
(in increasing order of starting times)
$\mathcal{O}_i=\{o_{i,1},o_{i,2},\ldots,o_{i,k_i}\}$.
Only the last interval, $o_{i,k_i}$, may end later than
$d(I_{i-1})$ (the ending time of $s_i$). 
If it does, then we merge $o_{i,1},o_{i,2},\ldots,o_{i,k_i-1}$ into 
one interval $pre_i$ such that $r(pre_i) = r(o_{i,1})$ and 
$d(pre_i) = r(o_{i,k_i})$, and thus
$p(pre_i) = r(o_{i,k_i})-r(o_{i,1}) \geq \sum^{k_i-1}_{j=1}
p(o_{i,j})$.  
By Lemma \ref{l1}, $f(p(pre_i)) \geq
\sum^{k_i-1}_{j=1} f(p(o_{i,j}))$. 
Otherwise, (i.e. $o_{i,k_i}$ ends before $d(I_{i-1})$), we merge all 
the intervals in $\mathcal{O}_i$ into one interval $pre_i$ such that
$r(pre_i) = r(o_{i,1})$ and
$p(pre_i) = d(o_{i,k_i}) - r(o_{i,1}) \geq \sum^{k_i}_{j=1}
p(o_{i,j})$.  
Thus, in both cases, such merging can only make $OPT$'s value larger.  
So we can assume that $OPT$ starts at most two intervals $pre_i$ and
$o_{i,k_i}$ in slot $s_i$.
After understanding the notations, we simply denote
the two intervals $pre_i$ and $o_{i,k_i}$ by $o_{i,1}$ and
$o_{i,2}$, respectively. 

The interval $o_{i,1}$ (if exist) is contained in $I_{i-1}$ and so
$p(o_{i,1}) \leq p(I_{i-1})$.
The interval $o_{i,2}$ (if exist) will end after $d(I_{i-1})$,
and $p(o_{i,2}) \leq p(I_{i})$ 
since $I_{i}$ is defined to be the longest interval that arrives
during slot $s_i$ and ends after $d(I_{i-1})$.
Note that $o_{i,2}$ may also end after $d(I_{i+l})$ for some $l\geq 0$. 
In this case, neither $o_{v,1}$ nor $o_{v,2}$ exist 
for $i \leq v \leq i+l$. 
If any $o_{i,1}$ or $o_{i,2}$ does not exist,
we set its length to zero. 

We now analyze the competitive ratio of $RAN$-$C$.
As in the monotone case, each phase can be analyzed separately.
Consider an arbitrary schedule
$S = \{I_0, I_1, \ldots, I_{n-1} \}$ produced by $RAN$-$C$ in a phase
with $n \ge 1$ slots, where $I_i$ overlaps $I_{i+1}$ $(0 \leq i < n-1)$, 
and the corresponding schedule
$S^* = \{o_{1,1}, o_{1,2}, o_{2,1}, \ldots, o_{n,1} \}$
produced by $OPT$ as $RAN$-$C$ produces $S$.
(Note that 
$o_{n,2}$ cannot
exist since otherwise this means there are some intervals that arrive 
within $[d(I_{n-2}),d(I_{n-1}))$ and end after $d(I_{n-1})$, and hence 
the phase will not end and $RAN$-$C$ will start an $I_n$.)

For each slot $i$, $OPT$ starts two intervals $o_{i,1}$ and
$o_{i,2}$ while $RAN$-$C$ accepts $I_{i-1}$. 
For presentation convenience,
let $x_{i,1}=p(o_{i,1})$, $x_{i,2}=p(o_{i,2})$
and $y_i=p(I_i)$. 
We already have that
$x_{i,1} \leq y_{i-1}$ for $1 \leq i \leq n$ 
and $x_{i,2} \leq y_{i}$ for $1 \leq i < n$.
We will show that 
\begin{equation} \label{eqn:main}
\sum^n_{i=1} f(x_{i,1}) + \sum^{n-1}_{i=1} f(x_{i,2})
   \leq \sum^{n-1}_{i=0} f(y_i).
\end{equation}

The left hand side of (\ref{eqn:main}) represents the total weight of 
intervals in $S^*$ (note that $o_{n,2}$ does not exist)
while the right hand side represents the total weight of intervals in
$S$. 
Since $RAN$-$C$ completes each interval in $S$ with probability 1/2,
its expected value is half of the right hand side of (\ref{eqn:main}).
Thus by proving (\ref{eqn:main}) we show the 2-competitiveness
of $RAN$-$C$.

We prove (\ref{eqn:main}) by induction on $n$.
When $n=1$, (\ref{eqn:main}) reduces to $f(x_{1,1}) \le f(y_0)$ 
which is true since $x_{1,1} \le y_0$.
Assume the claim holds for $n=k-1$, i.e.,
$\sum^{k-1}_{i=1} f(x_{i,1}) + \sum^{k-2}_{i=1} f(x_{i,2}) \leq 
\sum^{k-2}_{i=1} f(y_i)$.
Consider $I_k$, $o_{k-1,2}$ and $o_{k,1}$.
We have $x_{k-1,2} \le y_{k-1}$ and $x_{k,1} \le y_{k-1}$.
If $x_{k-1,2} + x_{k,1} \leq y_{k-1}$, then 
$f(x_{k-1,2})+f(x_{k,1}) \leq f(x_{k-1,2}+x_{k,1}) \leq f(y_{k-1})$.
Adding this to the induction hypothesis gives
$\sum^{k}_{i=1} f(x_{i,1}) + \sum^{k-1}_{i=1} f(x_{i,2}) \leq 
\sum^{k-1}_{i=0} f(y_i)$
and thus the claim holds for $n=k$.

Otherwise, if $x_{k-1,2}+x_{k,1} > y_{k-1}$, we first change the schedule
$S^*$ as follows: we increase the length of
$x_{k,1}$ to $y_{k-1}$ and decrease the length of $x_{k-1,2}$ by the
same amount.  The corresponding $r(o_{k-1,2})$ and $d(o_{k,1})$ are
fixed while both $d(o_{k-1,2})$ and $r(o_{k,1})$ decrease by an amount
of $y_{k-1}-x_{k,1}$.  
$OPT$ will only get better since
$f(x_{k,1}) + f(x_{k-1,2}) \leq f(y_{k-1}) + f(x_{k-1,2}-(y_{k-1}-x_{k,1}))$
by the properties of C-benevolent functions. 
After this change, $I_{k-1}$ and $o_{k,1}$ have the same length.
The new $o_{k-1,2}$ now ends on or before $d(I_{k-2})$.
We merge the new $o_{k-1,2}$ into $o_{k-1,1}$ so that the new $o_{k-1,1}$
extends its length to $x_{k-1,2}+x_{k-1,1}$ and keeps its start time
$r(o_{k-1,1})$ unchanged.  In the case that
$x_{k-1,1}=0$ before merging $o_{k-1,2}$, we set
$r(o_{k-1,1}) = r(o_{k-1,2})$.
The new $o_{k-1,1}$ is still contained by $I_{k-2}$
and thus $x_{k-1,1} \leq y_{k-2}$ still holds.
After merging, $x_{k-1,2}=0$ and $x_{k,1}=y_{k-1}$.  
Therefore
$\sum^{k}_{i=1} f(x_{i,1}) + \sum^{k-1}_{i=1} f(x_{i,2}) 
= \sum^{k-1}_{i=1} f(x_{i,1}) + \sum^{k-2}_{i=1} f(x_{i,2}) + f(x_{k,1})
\leq \sum^{k-2}_{i=0} f(y_i) + f(y_{k-1})
= \sum^{k-1}_{i=0} f(y_i)$.
Thus the claim is true for $n=k$.  
\end{proof}

\subsection{D-benevolent Instances}

\paragraph{Algorithm $RAN$-$D$.}
The basic idea of $RAN$-$D$ is same as $RAN$: 
two algorithms $A$ or $B$ are executed each with probability 1/2.
Intuitively, in an odd slot (where slots will be defined precisely 
in the following paragraphs), 
$A$ accepts the largest-weight interval arriving during that slot, 
by starting an interval
and preempting if a new one arrives with a larger weight. We call the
interval being executed by $A$ the {\it main interval}, denoted by $I_M$. 
Meanwhile, $B$
continues to run to completion the interval started in the previous slot;
we call this the {\it residual interval}, denoted by $I_R$.
This residual interval must be completed 
(as in the equal length case) because this is the interval accepted in the
previous slot. However in the D-benevolent case, 
if a shorter (and therefore larger weight) interval arrives, 
the residual interval can actually be preempted and replaced by this new
interval.  
For even slots the roles of $A$ and $B$ are reversed (and the interval
started by $B$ is the main interval and the one completed by $A$ the
residual interval).

Unlike $RAN$-$M$ or $RAN$-$C$, here when slot $s_{i-1}$ finishes, 
the next slot $s_i$ is not completely determined: 
slot $s_i$ begins where $s_{i-1}$ ends,
but the ending time of slot $s_i$ will only get a provisional value, 
which may become smaller (but not larger) later on. 
This is called the {\it provisional ending time} of the slot, 
denoted by $e_i$.
Slots will also be grouped into phases as in the other types of instances.

Note that $I_M$, $I_R$ and $e_i$ change during the execution of the
algorithm, even within the same slot. 
But $RAN$-$D$ always maintains the following invariant:

\begin{quote}
{\bf Invariant:} 
Suppose $I_R$ and $I_M$ are the residual and main interval respectively
during execution in a slot $s_i$.
Then $e_i = d(I_R) \le d(I_M)$ (if the intervals exist).
Moreover $e_i$ can only be decreased, not increased.
\end{quote}

We describe the processing of intervals in a slot $s_i$ ($i \ge 1$).
Consider an odd slot $s_i$ (the case of even slots is the same with the
roles of $A$ and $B$ reversed).
At the beginning of $s_i$, $A$ is idle and 
$B$ is continuing the execution of a residual interval $I_R$. 
At this point $e_i$ is provisionally set to $d(I_R)$.
In the case of the first slot, there is no residual interval left over
from the previous slot, so we set $e_i$ to be the deadline of the first
interval that arrives. If more than one interval arrive at the same instant,
choose anyone.
 
Consider a time during $s_i$ when an interval $I$ arrives while
$A$ and $B$ are respectively executing some intervals $I_M$ and $I_R$.
If more than one interval arrive at the same instant, process them in any
order.
If $A$ or $B$ is idle, assume $I_M$ or $I_R$ to have weight 0.
Then $A$ and $B$ react according to the following three cases:

\begin{enumerate}
\item
If $d(I) \ge e_i$ and $w(I)>w(I_M)$,
then $I$ preempts $I_M$, and this $I$ becomes the new $I_M$.
In this case, $e_i$ remains unchanged.
\item
If $d(I) < e_i$ (which implies $w(I) \ge w(I_M)$ and $w(I) \ge w(I_R)$
because by the invariant, $d(I_M) \ge d(I_R) = e_i > d(I)$, 
and $I$ arrives no earlier than either $I_M$ or $I_R$, 
and thus $I$ is shorter),
then $I$ preempts both $I_M$ in $A$ and $I_R$ in $B$.
Here $I$ becomes the new $I_M$ and $I_R$, and
$e_i$ is then set to $d(I)$.  
\item
Otherwise, $w(I) \le w(I_M)$ and $I$ is discarded.
\end{enumerate}

Observe that the invariant is always maintained when we change any of $I_M$,
$I_R$ or $e_i$.

This process repeats until time $e_i$ is reached and slot $s_i$ ends.
If $d(I_M)>e_i$ at the end of slot $s_i$, then a new slot $s_{i+1}$ begins 
where slot $s_i$ ends.
$A$ has not finished execution of $I_M$ yet, so it
now becomes the $I_R$ of slot $s_{i+1}$, and
$e_{i+1}$ is provisionally set to $d(I_R)$.
Otherwise, $d(I_M)=e_i$ and $A$ just finishes execution of $I_M$, 
then the phase ends. 
In this case we wait until the next interval 
arrival, then a new phase starts.

Note that $RAN$-$D$ needs to simulate the execution of both $A$ and $B$
(to determine when slots end) but the actual execution follows only one 
of them.

\begin{theorem}
$RAN$-$D$ is $2$-competitive for online interval scheduling on
D-benevolent instances.
\end{theorem}
\begin{proof}
Consider each slot $s_i = [e_{i-1}, e_i)$.
We claim that $OPT$ can start at most one interval in $s_i$
and that this interval cannot finish strictly before $e_i$.
The first part of the claim follows from the second since
if $OPT$ starts two or more intervals within $s_i$, then the first 
such interval must end strictly before $e_i$. 
Assume to the contrary that $OPT$ starts an interval $I$ that finishes
strictly before $e_i$. 
Then $I$ also finishes strictly before the provisional value of $e_i$ 
at the moment $I$ arrives, since the provisional ending time only decreases.
By the design of the algorithm, at that point
$e_i$ will be reduced to $d(I)$. 
$e_i$ may be reduced further subsequently, but in any case this contradicts
the fact that $d(I)<e_i$.
Hence the claim follows.

Now suppose $OPT$ starts an interval $I$ in an odd slot $s_i$ and eventually
completes it.
We will show that if $s_i$ is not the last slot in the phase, 
$A$ will complete an interval of weight no less than $w(I)$ in slot $s_{i+1}$;
if $s_i$ is the last slot, then $A$ will complete an interval of weight
no less than $w(I)$ in slot $s_i$.

Consider the moment when $I$ arrives in $s_i$.
If $I$ has larger weight than the current $I_M$,
$A$ will preempt it and start $I$.
Thus, by the end of $s_i$,
$A$ should have started a main interval $I_M$ of weight at least $w(I)$.
If this is the last slot, then $A$ completes $I_M$ at the end of $s_i$.
Otherwise, $I_M$ becomes the residual interval in slot $s_{i+1}$
and $A$ will execute it to completion (as an residual interval) in $s_{i+1}$
unless another interval $I'$ arrives in $s_{i+1}$ such that 
$d(I') < e_{i+1}$ (and hence $w(I') \ge w(I_M)$).
Note that $e_{i+1}$ will then be reduced to $d(I')$.
This $I'$ may still be preempted by intervals of even larger weight
and earlier deadline.
In any case, at exactly the end of the next slot $s_{i+1}$,
$A$ would have completed the residual interval.

We can make a similar claim for even slots.
Therefore it follows that, for every interval started by $OPT$, 
either $A$ or $B$ will complete an interval of at least the same weight 
in the same or the next slot.
Thus the total value of $A$ and $B$ is no less than that of $OPT$.
The 2-competitiveness then follows since each of A/B is executed with
1/2 probability.  
\end{proof}

\section{Algorithms for Equal Length Jobs}

\paragraph{Algorithm $RAN$-$J$.}
In this section we extend $RAN$ to the online scheduling of 
equal length jobs with restarts.
The algorithm remains very simple but the analysis is more involved.
Again $RAN$-$J$ chooses between two deterministic algorithms $A$ and $B$, 
each with probability 1/2, and again
$A$ takes care of odd slots and $B$ takes care of even slots,
where the slots are defined as in the equal length interval case
(i.e. they all have unit length).  
At the beginning of each odd slot, $A$ considers all pending jobs that can 
still be completed, and starts the one with the largest weight.
(If there are multiple jobs with the same maximum weight, start an 
arbitrary one.)
If another job of a larger weight arrives within the slot, 
$A$ aborts the current job and starts the new one instead.  
At the end of this odd slot, the job that is being executed will run to
completion (into the following even slot) without abortion.
$A$ will then stay idle until the beginning of the next odd slot.
Even slots are handled by $B$ similarly.

The following simple example (see Figure \ref{fig:ran-j}(a))  
illustrates the algorithm, and shows that 
$RAN$-$J$ is not better than 3-competitive. 
Consider three jobs $X, Y, Z$, where $r(X)=0, d(X)=3, w(X)=1+\epsilon$ for
arbitrary small $\epsilon>0$;
$r(Y)=0, d(Y)=1, w(Y)=1$; and
$r(Z)=1, d(Z)=2, w(Z)=1$.
Both $A$ and $B$ will complete $X$ only, but $OPT$ can complete all three.

\paragraph{Notations.}
We define some additional notations that will be used in the rest of 
this section to make our discussion clearer.
The notation $[s_1..s_2]$ denotes a range of slots 
from slot $s_1$ to $s_2$ inclusive, where $s_1$ is before $s_2$.
Arithmetic operators on slots carry the natural meaning,
so $s+1$ is the slot immediately after $s$,
$s-1$ is the slot immediately before $s$, $s_1 < s_2$
means $s_1$ is before $s_2$, etc.
The job accepted by an algorithm $A$ in slot $s$ is denoted by $A(s)$.
(Any algorithm can accept at most one job in each slot since the slot has
the same length as a job.)
We define the inverse $A^{-1}(x)$ to be the slot $s$ with $A(s)=x$, 
if it exists; otherwise it is undefined.

\paragraph{Charging scheme.}
Our approach to the proof is to map (or charge) the weights of jobs 
accepted by $OPT$ to slots where $A$ or $B$ have accepted `sufficiently
heavy' jobs; namely, that each slot $s$ 
receives a charge at most 1.5 times of $w(A(s))$ or $w(B(s))$.
In some cases this is not possible and we pair up slots with large charges
with slots with small charges so that the overall ratio is still at most
1.5.  Since each job in $A$ or $B$ is completed with probability 1/2 only,
the expected value of the online algorithm is half the total value of 
$A$ and $B$.  This gives a competitiveness of 3.

The charging scheme is defined as follows.
Consider a slot $s$ where $OPT$ accepts the job $OPT(s)$.
Suppose $s$ is odd (so $A$ is choosing the heaviest job to start).
If $w(A(s)) \ge w(OPT(s))$, charge the weight of $OPT(s)$ to $s$.
We call this a {\it downward charge}.
Otherwise, $A$ must have 
accepted $OPT(s)$ at some earlier slot $s'$.
Charge half the weight of $OPT(s)$ to this slot $s'$.
This is called a {\it self charge}.
For $B$, either it has accepted the job $OPT(s)$ before $s$, 
in which case we charge the remaining half to that slot
(this is also a self charge);
or $OPT(s)$ is still pending at slot $s-1$, which means at slot $s-1$, 
$B$ accepts a job with weight at least $w(OPT(s))$.
Charge the remaining half to the slot $s-1$.
This is called a {\it backward charge}.
When $s$ is an even slot the charges are similarly defined.

\begin{figure}
\centerline{ \epsfysize=2in \epsffile{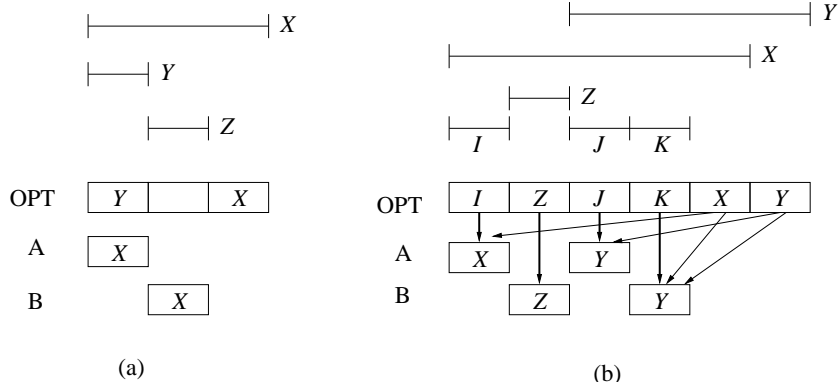} }
\caption{(a) An example showing $RAN$-$J$ is not better than 3-competitive.
(b) An example showing the charges and a bad slot.
The weight of the jobs are (for small $\epsilon>0$):
$w(X)=1+\epsilon, w(Y)=1+2\epsilon, w(Z)=1+3\epsilon$;
$w(I)=w(J)=w(K)=1$. Slot 4 is a bad slot.}
\label{fig:ran-j}
\end{figure}

Clearly, all job weights in $OPT$ are charged to some slots.
Observe that for each charge from $OPT$ to a slot,
the weight of the job generating the charge is no more than
that of the job accepted in the slot receiving the charge.
We define each downward charge to be of one {\it unit}, and each self
or backward charge to be of 0.5 unit.
With this definition, if every slot receives at most 1.5 units of charge, 
then we are done.
Unfortunately, slots can receive up to 2 units of charges
because a slot can receive at most one charge of each type.
Slots receiving 2 units of charges are called {\it bad}; they must receive
a backward charge.
Slots with at most 1 unit charge are called {\it good}.  
Each bad slot $s$ can be characterized by a pair $(X,Y)$ where $X$ is the 
job $A(s)$ or $B(s)$, and $Y$ is the job $OPT(s+1)$ generating the backward
charge.
The example in Figure \ref{fig:ran-j}(b) illustrates the charges 
and the existence of bad slots.

\paragraph{Competitive Analysis.}

The key part of the proof is to deal with bad slots.
For each bad slot, we pair it up with a good slot 
so that the `overall' charge is still under a ratio of 1.5.
The proof of the following lemma will show how this is done. 

\begin{lemma} \label{lem:match-slots}
For each bad slot $s=(X,Y)$, there is a good slot $s'$ such that 
the weight of $A(s')$ or $B(s')$ is at least $w(Y)$.  
Moreover, any two bad slots are paired with different good slots.
\end{lemma}

If Lemma~\ref{lem:match-slots} is true, then we have

\begin{lemma}
Slots $s$ and $s'$ as defined in Lemma~\ref{lem:match-slots} together 
receive a charge at most $1.5$ times the total weight of the jobs in $A/B$
in the two slots.
\end{lemma}
\begin{proof}
Let $w_s$ and $w_{s'}$ be the weight of jobs accepted by $A/B$ in $s$ and $s'$
respectively.
The charges to $s$ is at most $1.5w_s + 0.5w(Y)$ while the charges to $s'$
is at most $w_{s'}$.  The overall ratio is therefore
$(1.5w_s + 0.5w(Y) + w_{s'})/(w_s + w_{s'}) \le 1.5$
since $w(Y) \le w_{s'}$. 
\end{proof}

\begin{theorem}
$RAN$-$J$ is $3$-competitive for the online scheduling of equal length jobs 
with restarts.
\end{theorem}
\begin{proof}
All the weights of jobs accepted by $OPT$ 
are charged to slots in $A$ or $B$.
Each slot in $A$ and $B$ receives charges at most 1.5 times the weight
of the job in the slot,
either as a single slot or as a pair of slots as defined in 
Lemma~\ref{lem:match-slots}.  
Since each job in $A$ or $B$ is only completed 
with probability 1/2, the expected value of the online algorithm is half the
total weight of jobs in $A$ and $B$.  
It follows that the competitive ratio is 3. 
\end{proof}

Before proving Lemma~\ref{lem:match-slots}
we first show some properties of bad slots in the following
lemma.  (Although the lemma is stated in terms of odd slots,
the case of even slots is similar.)

\begin{lemma} \label{lem:badslot}
For each bad slot $s = (X,Y)$, where $s$ is an odd slot,
\begin{description}
\item[(i)]
both $X$ and $Y$ are accepted
by $B$ in some slots before $s$ (call the slots $s_0$
and $s_1$, where $s_0 < s_1$); and
\item[(ii)]
for each odd slot $\tilde{s}$ in $[s_0+1..s_1-1]$,
$w(A(\tilde{s})) \geq w(B(s_0)) \ge w(Y)$, and
for each odd slot $\tilde{s}$ in $[s_1+1..s-1]$,
$w(A(\tilde{s})) \geq w(B(s_1)) \ge w(Y)$.
\end{description}
\end{lemma}
\begin{proof}
(i) Since slot $s+1$ makes a backward charge
instead of a downward charge, we have $w(B(s+1)) < w(Y)$.  
Hence $B$ must have accepted $Y$ before $s$,
or else $Y$ could have been a candidate for $B(s+1)$.
Furthermore, $w(X) \ge w(Y) > w(B(s+1))$.
By the same reasoning, $B$ must have accepted $X$ before $s$.

\noindent (ii)
If $B(s_0)=Y$, then $Y$ has already arrived before the end of 
slot $s_0$ but is not accepted by $A$ at/before $s$.
Hence $A$ must have accepted jobs with weights at least $w(Y)$
in all odd slots in [$s_0+1..s_1-1$].
If $B(s_0)=X$ then the same reasoning implies that 
$A$ accepted jobs with weights at least $w(X)$, which is at
least $w(Y)$, in these slots.
The same argument holds for $A$[$s_1+1..s$]. 
\end{proof}

We now prove Lemma~\ref{lem:match-slots}.
We give a step-by-step procedure for identifying a good slot 
(in which $A$ or $B$ has accepted a job of sufficient weight) 
for every bad slot.
Consider an odd bad slot $s=(X,Y)$. (The case for even slots
is similar.)  
Roughly speaking, the procedure initially identifies the two slots
$s_0$ and $s_1$ defined in Lemma \ref{lem:badslot}
and designates $s_1$ as a {\em special slot}, denoted by $s^*$.
Then it checks if $s^*$ or $s^*-1$ is a good slot.
If a good slot is found, the procedure stops.
Otherwise, it will identify a new slot not found before,
pick a new special slot $s^*$ from among the
identified slots;
and then move to the next step (which checks on $s^*$, $s^*-1$
and so on).

In more detail, at the beginning of step $i$ ($i \geq 1$),
a collection of $i+1$ slots, $s_0 < s_1 < \cdots < s_i$,
have been identified.
They are all even slots before the bad slot $s$ and
one of them is designated as the special slot $s^*$.
Denote by $Y_j$ the job $B(s_j)$ for all
$j \in \{0, \ldots, i\}$ and for convenience, let
$s_{i+1}$ denote $s$.
Step $i$ proceeds as follows:
\begin{description}
\item[{\rm Step $i$.1.}] 
Consider the job $Y^* = B(s^*)$ in slot $s^*$.
By Lemma \ref{lem:A-weights}(i) below,
$Y^*$ has weight at least $w(Y)$.
So, if the slot $s^*$ receives at most 1 unit of charge,
then we have identified a good slot of sufficient weight
and we stop.

\item[{\rm Step $i$.2.}] 
Otherwise, $s^*$ has at least 1.5 unit of charge
and must therefore have a downward charge.
Denote by $Z$ the job $A(s^*-1)$.
By Lemma \ref{lem:A-weights}(ii) below, $w(Z) \geq w(Y)$.
Since slot $s^*$ must have a downward charge,
slot $s^*-1$ cannot receive a backward charge.
If slot $s^*-1$ receives no self charge as well, 
then it is a good slot and we are done.  

\item[{\rm Step $i$.3.}]
Otherwise $s^*-1$ receives a self charge and hence $Z$ is
accepted by $OPT$ in some slot $s'$ after $s^*$.
In Lemma \ref{lem:case3}, we will show that $B$ must also
accept $Z$ at a slot $s''$ where $s'' < \min\{s, s'\}$.
Note that $Z$ is not in $\{ Y_0, Y_1, \ldots, Y_i \}$.
(A job in  $\{ Y_0, Y_1, \ldots, Y_i \}$ is either the
job $Y$, which is not accepted by $A$ before slot $s$,
or a job accepted by $A$ in a slot other than
$s^*-1$.)
Therefore, $s''$ is a different slot than 
$s_0, s_1, \ldots, s_i$.

Mark slot $s_0$ or $s''$, whichever is later,
as the new special slot $s^*$.
Re-index $s_0, \ldots, s_i$ and $s''$
as $s_0 < s_1 < \cdots < s_{i+1}$ and
move on to Step ($i+1$).
\end{description}

We need to show that
(i) the procedure always terminates,
(ii) the claims made in the above procedure are correct, and
(iii) any two bad slots are paired with different good slots
following this procedure.  
The first is easy: note that in each step, 
if a good slot is not found, a new slot, $s''$,
which is before $s$, is identified instead.
But there are only a finite number of slots before $s$. 
Therefore, the procedure must eventually terminate and
return a good slot.

The claim in Step $i$.1 and $i$.2 is proved in the lemma below,
which is basically a generalization of Lemma~\ref{lem:badslot}.

\begin{lemma} \label{lem:A-weights}
For any step $i$ ($i \geq 1$) and any $j \in \{0, \ldots, i\}$, 
\begin{description}
\item[(i)]
$w(Y_j) \geq w(Y)$
and
\item[(ii)]
for all odd slots $\tilde{s}$ in $[s_j+1 \ldots s_{j+1}-1]$,
$w(A(\tilde{s})) \geq w(Y_j)$.
\end{description}
\end{lemma}

\begin{proof}
The proof is by induction on $i$.
Clearly, (i) and (ii) are true for $i=1$
as proved by Lemma \ref{lem:badslot}.
Suppose (i) and (ii) are true at the beginning
of some step $i$.
We will show that they are maintained
at the beginning of step $i+1$.

Recall that $Z$ is accepted by $A$ in slot
$s^*-1$ and by $B$ in slot $s''$.
By (ii), $w(Z) \geq w(Y)$.
Thus, $w(B(s'')) \geq w(Y)$ and hence
(i) is maintained in the next step.

To show that (ii) is also maintained in the next step,
it suffices to show that for any odd slot
$\tilde{s}$ in $[s''+1 \ldots s_q -1]$
where $s_q$ is the closest slot among
$s_0, s_1, \ldots, s_i$ after $s''$
(or $s_q = s$ if $s''$ lies after $s_i$), 
$w(A(\tilde{s})) \geq w(Z)$.
We consider two cases:

If $s'' < s^*-1$,
then $Z$ is available before the end of slot $s''$ and yet
is not accepted by $A$ until $s^*-1$.
See Figure \ref{fig:induction}(a).
Therefore, $w(A(\tilde{s})) \geq w(Z)$
for every $\tilde{s} \in [s''+1..s^*-1]$.

\begin{figure}[h]
\centerline{ \epsfysize=5cm \epsffile{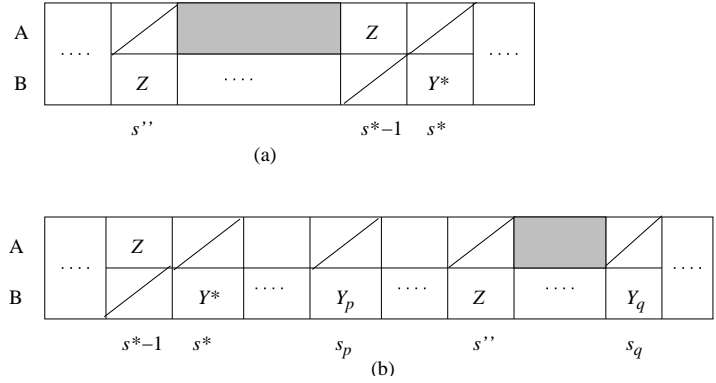} }
\caption{Positions of $Y_{i+1}$ and $Y_p$ in $A$ and $B$.}
\label{fig:induction}
\end{figure}

If $s'' > s^*-1$,
then let $s_p$ be the closest slot among
$s_0, s_1, \ldots, s_i$ before $s''$.
See Figure \ref{fig:induction}(b).
Such slot must exist because $s^*$ is
one such candidate.
Then $w(Y_p) \geq w(Z)$ or else $Z$ would
have been accepted in slot $s_p$.
Therefore, $w(A(\tilde{s})) \geq w(Y_p) \geq w(Z)$
for every odd slot $\tilde{s} \in [s''+1..s_q-1]$.  
\end{proof}

\begin{lemma} \label{lem:case3}
In Step $i$.3, $B$ accepts $Z$ in a slot $s''$ where $s'' < \min(s,s')$.
\end{lemma}
\begin{proof}
First notice that $w(B(s+1)) < w(Y) \leq w(Y_{i+1})$.
Therefore, $s'' < s$
or else $Z$ would have been a candidate for slot $s+1$.

Now we assume that $s' < s$ and show that $s'' < s'$.
We distinguish two cases.

Case i.  $s'$ is an odd slot.  
Let $U = A(s')$.
Let $s_p$ be the slot in $s_0, \ldots, s_i$
that is closest to and before $s'$ 
(which must exist because $s^*$ itself is a candidate).
Note that $s^*-1$ must be before $s_p$ since it must be 
before $s'$ and immediately before one of the $s_j$'s
(in this case $s^*$), and $s_p$ is the latest such
$s_j$'s before $s'$.  

\begin{figure}[h]
\centerline{ \epsfysize=1in \epsffile{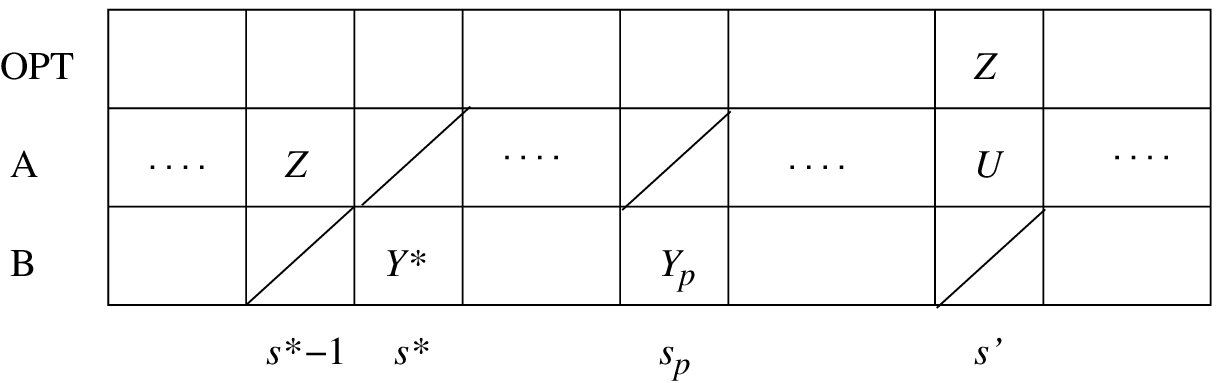} }
\label{fig:s3}
\end{figure}

We have $w(Z) > w(U)$ since a self charge is made instead of a 
downward charge.
By Lemma \ref{lem:A-weights}(ii), $w(U) \ge w(Y_p)$.  
Therefore $w(Z) > w(Y_p)$.
Hence $Z$ must be accepted before $Y_p$ 
in $B$ or else it can take $Y_p$'s place in $B$.
By definition of $s_p$, $s_p < s'$.
Hence $s'' < s'$.

Case ii.  $s'$ is an even slot. 
Let $U = B(s')$.
Then $w(Z) > w(U)$ due to no downward charge in slot $s'$.  
Thus $B$ must have accepted
$Z$ before $s'$, i.e., $s'' < s'$.

So in all cases $B$ accepts $Z$ at some slot $s'' < \min(s, s')$. 
\end{proof}

Finally, the lemma below shows that
two bad slots are paired with different good slots.

\begin{lemma} \label{lem:distinct}
All bad slots are paired with different good slots.
\end{lemma}
\begin{proof}
There are two possible places in our procedure where good slots can be
identified: in Step $j$.1 or in Step $k$.2 for some $j$ and $k$.
Call them substeps 1 and 2.
Note that, for an odd bad slot, good slots identified in substep 1 are 
always when $B$ is accepting jobs, and good slot identified in substep 2 are 
always when $A$ is accepting jobs, and vice versa for even bad slots.

Consider two distinct bad slots $s=(X,Y)$ and $s'=(X',Y')$.
First, consider the case when $s$ and $s'$ has different parity
(odd or even slots). 
Then they can match with the same good slot only if one of them identifies it 
in substep 1 and the other in substep 2.
However, a good slot in substep 1 must receive self-charge (this is how the 
$Y_j$'s are identified) while a good slot in substep 2
cannot receive a self charge (otherwise we would have moved on to some 
Step $i$.3 in the procedure).  
Thus it is impossible that a substep 1 good slot is
also a substep 2 good slot.

Next, consider the case when $s$ and $s'$ are of the same parity.
Without loss of generality assume that they are both odd slots.
To facilitate our discussion, we re-index the $Y_j$'s in the order
they are identified.
So we let $Y_0, Y_1, Y_2, \ldots$ be the chain of $Y_j$'s associated with $s$
where $\{Y_0, Y_1\} = \{X, Y\}$
and for $j \geq 2$, $Y_j$ is the job identified in step $(j-1)$.3.
Similarly, we let $Y'_0, Y'_1, Y'_2 \ldots$ be the chain associated with $s'$.
We will show that
no job appears in both chains.
This proves the claim because
for two bad slots of the same parity to be matched to the same good slot,
they must both be identified in substep 1 or both in substep 2.
But if the chains of $Y_i$'s associated with them are different,
this is not possible.

To show the chains are distinct,
we first show that $Y_0, Y_1, Y'_0$ and $Y'_1$ are all distinct.
Recall that $\{X,Y\}=\{Y_0,Y_1\}$ and $\{X',Y'\}=\{Y'_0,Y'_1\}$.
Clearly $X \neq Y, X' \neq Y', X \neq X'$ and $Y \neq Y'$.
Thus we only need to show that $X \neq Y'$ and $X' \neq Y$.
If this is not true, then either:
(1) $Y$ is accepted in $A$ in a bad slot; or
(2) $X$ in $OPT$ generates a backward charge.  
For (1), if $A^{-1}(Y) < A^{-1}(X) (=s)$, then $OPT(s+1) (=Y)$ 
would not make a backward charge to $s$; while if
$A^{-1}(Y) > s$, then $A^{-1}(Y)$ cannot get a self 
charge and hence receives at most 1.5 units of charge.
For (2), $OPT^{-1}(X) > OPT^{-1}(Y)$ due to the self charge to slot $s$, and 
by Lemma~\ref{lem:badslot}(i), $X$ must also be accepted in $B$ before $s$.  
Hence $X$ cannot generate a backward charge.
Thus neither (1) nor (2) can be true.

We have now established that $Y_0, Y_1, Y'_0$ and $Y'_1$ are all different.
It is also clear that $Y_j \neq Y_k$ and $Y'_j \neq Y'_k$ for any $j$ and $k$.
Note that, if $Y_j = Y'_k$, for some $j>1$ and $k>1$, 
then there must be some $j'<j$ and $k'<k$ such that $Y_{j'} = Y'_{k'}$,
because they are uniquely defined in such a way (in some substep 2). 
So the only remaining case to consider is $Y_0$ or $Y_1$ being the same 
as $Y'_j$ for some $j>1$. 
Recall $Y_0$ and $Y_1$ are the $X$ and $Y$ of $s$. 
$Y'_j$ cannot be $Y$
because by Lemma \ref{lem:case3},
both $A$ and $B$ accept $Y'_j$ before $OPT$ does
but $A$ accepts $Y$ after $OPT$. 
$Y'_j$ cannot be $X$ because this would mean the job $B(s+1)$ is $Y'_k$ for 
some $k<j$, so slot $s+1$ should receive a downward charge 
but this contradicts that $OPT(s+1)$ makes a backward charge to $s$ 
instead of a downward charge. 
\end{proof}

\section{Lower Bound for Equal Length Intervals}
  \label{sec:LB}

In this section, we show a lower bound of 2 for barely
random algorithms for scheduling equal length intervals
that choose between two deterministic algorithms,
possibly with unequal probability.

\begin{theorem}
No barely random algorithm choosing between two deterministic algorithms
for the online scheduling of equal length intervals
has a competitive ratio better than $2$.
\end{theorem}

Let $ALG$ be a barely random algorithm that chooses between 
two deterministic algorithms $A$ and $B$ with 
probability $p$ and $q$ respectively
such that $p+q=1$ and $0 < p \leq q < 1$.
Let $\delta$ be an arbitrarily small positive constant
less than 1.
We will show that there is an input on which $OPT$
gains at least $2-\delta$ times of what $ALG$ gains.

We will be using sets of intervals similar to
that in Woeginger \cite{Woeg94}.
More formally, let $\epsilon$ be an arbitrary positive real number
and let $v, w$ be any pair of real numbers such that
$0 \leq v \leq w$. 
We define $SET(v,w,\epsilon)$ as a set of intervals of weight 
$v, v+\epsilon', v+2\epsilon' , \dots, w$ 
(where $\epsilon'$ is the largest number 
such that $\epsilon' \leq \epsilon$ and 
$w-v$ is a multiple of $\epsilon'$)
and their relative arrival times are such that
intervals of smaller weight come earlier
and the last interval
(i.e., the one that arrives last and has weight $w$)
arrives before the first interval finishes.
Thus, there is overlapping between any pair of
intervals in the set.
See Figure~\ref{fig:set}.
This presents a difficulty for the online algorithm
as it has to choose the right interval to process
without knowledge of the future.

\begin{figure}
\centerline{ \epsfysize=1.2in \epsffile{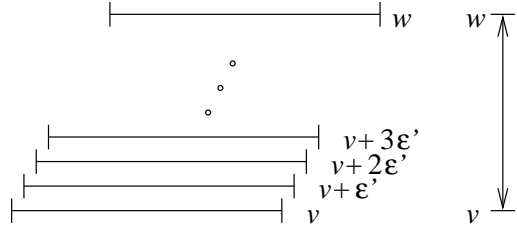} }
\caption{$SET(v,w,\epsilon)$.  On the left is the actual set of 
intervals; the vertical arrow on the right is the notation we use to denote 
such a set.}
\label{fig:set}
\end{figure}

To facilitate our discussion, 
we assume that all intervals have weight at least 1
throughout this section, except Section \ref{subsec:ylarge}.
If $I$ is an interval in $SET(v, w, \epsilon)$ and $w(I) > v$
(i.e., $I$ is not the earliest interval in the set),
then $I^-$ denotes the interval
that arrives just before $I$ in $SET(v, w, \epsilon)$.
So, $w(I^-) \geq w(I) - \epsilon$.

\subsection{A Few Simple Cases}

We first present a few simple situations in which $OPT$ can gain
a lot compared with what $ALG$ can gain.
The first lemma shows that an algorithm should not start an interval
that is lighter than the current interval being processed by the
other algorithm.
The second lemma shows that it is not good to have
$A$ processing an interval of equal or heavier weight
than the interval currently being processed by $B$.
Moreover, the two algorithms should avoid processing
almost non-overlapping intervals
as shown in the third lemma.

\begin{lemma} \label{lem:YEarly}
Suppose at some moment,
one of the algorithms (say $B$) is processing
an interval $J$ from a set $S = SET(v,w,\epsilon)$
while the other algorithm ($A$) is processing
another interval $I$,
where $w(J)>v$, $w(I) \leq w(J)$ and $r(I)>r(J)$.
(Note that here the interval $I$ cannot come from $S$.)
Then $OPT$ gains at least $2-\epsilon$ of $ALG$'s
gain on an input consisting of
$S$, $I$ and some subsequently arrived intervals.
\end{lemma}

\begin{proof}
We illustrate the scenario in Figure \ref{fig:yearly}.
(There, the vertical line represents the set $S$ 
and the horizontal line labelled $J$ is one of the intervals
in $S$.
The horizontal line labelled $I$ arrives later than $J$
and has a smaller weight.)
To defeat $ALG$, 
an interval $J'$ with the same weight as $J$ is released
between $d(J^-)$ and $d(J)$
and no more intervals are released.
(See Figure \ref{fig:yearly}.)
Then $OPT$ completes $J^-$ (which finishes just before $J'$ starts)
and $J'$, gaining $w(J^-) + w(J')$ $\geq 2w(J) - \epsilon$
while $ALG$ gains at most $(p+q) w(J) = w(J)$ even if
it aborts $I$ to start $J'$.
Since $w(J) \geq 1$, $\frac{2w(J)-\epsilon}{w(J)} \geq 2-\epsilon$.
\end{proof}

\begin{figure}
\centerline{ \epsfysize=0.9in \epsffile{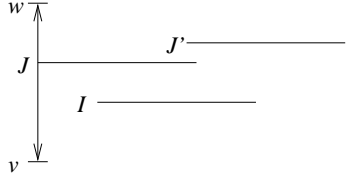} }
\caption{(Lemma \ref{lem:YEarly}) Starting a lighter interval later}
\label{fig:yearly}
\end{figure}

\begin{lemma} \label{lem:AHeavy}
Suppose at some moment, 
algorithm $B$ is processing an interval $I$ from
a set $S = SET(v, w, \epsilon)$
while $A$ is processing an interval $J$ from
a set $S' = SET(v', w', \epsilon')$,
where $w(I)>v$, $w(J)>v'$, $w(I) \leq w(J)$
and $r(I) \leq r(J)$.
(Note that $S$ and $S'$ can be the same set;
$I$ and $J$ can even be the same interval.)
Then $OPT$ gains at least $2-\max\{\epsilon, \epsilon'\}$
of $ALG$'s gain on an input consisting of $S$, $S'$
and some subsequently arrived intervals.
\end{lemma}

\begin{proof}
An interval $J'$ with the same weight as $J$ is released
between $d(I^-)$ and $d(I)$.
See Figure \ref{fig:aheavy}.
Clearly there is no point in algorithm $A$ aborting $J$ to
start $J'$.
If algorithm $B$ continues with $I$,
then no more intervals arrived.
$OPT$ gains $w(I^-) + w(J') \geq w(I) + w(J) - \epsilon$
while $ALG$ gains $q w(I) + p w(J)$ $\leq (w(I) + w(J))/2$.
Since $(w(I)+w(J))/2 \geq 1$,
we have $\frac{w(I)+w(J)-\epsilon}{(w(I)+w(J))/2} \geq 2-\epsilon$.
If $B$ aborts $I$ and starts $J'$,
then using Lemma \ref{lem:YEarly}, we can see that
$OPT$ gains at least $2-\epsilon'$ of what $ALG$ gains
on an input consisting of $S$, $S'$ and the subsequently
arrived intervals specified in Lemma \ref{lem:YEarly}.
\end{proof}

\begin{figure}
\centerline{ \epsfysize=0.9in \epsffile{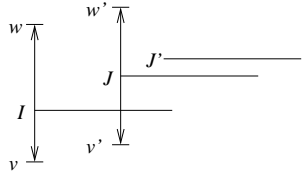} }
\caption{(Lemma \ref{lem:AHeavy}) $A$ processing a heavier interval ($J$)}
\label{fig:aheavy}
\end{figure}

\begin{lemma} \label{lem:ABSeparate}
Suppose at some moment, 
algorithm $A$ is processing an interval $I$ from
a set $S = SET(v, w, \epsilon)$ 
while algorithm $B$ is processing an interval $J$ from
another set $S' = SET(v', w', \epsilon')$, where $w(I)>v$,
$w(J)>v'$, $w(I) \leq w(J)$
and the intervals of $S'$ arrive between $d(I^-)$ and $d(I)$.
Then $OPT$ gains at least $2-(\epsilon+\epsilon')$ of $ALG$'s
gain on an input consisting of
$S$, $S'$ and some subsequently arrived intervals.
in the worst case.
\end{lemma}

\begin{proof}
An interval $J'$ with the same weight as $J$ is
released between $d(J^-)$ and $d(J)$
and no more intervals are released.
See Figure \ref{fig:abseparate}.
OPT completes $I^-$ in $S$, $J^-$ in $S'$ and $J'$.
So it gains $w(I^-) + w(J^-) + w(J')$
$\geq w(I) + 2w(J) - (\epsilon+\epsilon')$.
On the other hand, $A$ completes $I$ and then $J'$
while $B$ completes $J$.
Thus $ALG$ gains at most $p(w(I)+w(J)) + qw(J) = pw(I) + w(J)$
$\leq (1/2)w(I) +w(J)$.
Since $w(I)/2 + w(J) \geq 1$, we have
$\frac{w(I)+2w(J)-(\epsilon+\epsilon')}{w(I)/2+w(J)}$
$\geq 2-(\epsilon+\epsilon')$.
\end{proof}

\begin{figure}
\centerline{ \epsfysize=0.9in \epsffile{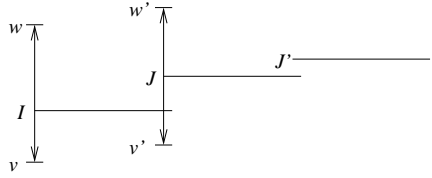} }
\caption{(Lemma \ref{lem:ABSeparate}) $A$ and $B$ processing
almost non-overlapping intervals}
\label{fig:abseparate}
\end{figure}

\subsection{Constructing the Sequence of Intervals}

Our lower bound proof takes a number of steps.
In each step, the adversary will release some set of intervals
$SET(\cdot,\cdot,\cdot)$ adaptively according to how $ALG$ reacts
in the previous steps.
In each step, the adversary forces $ALG$ not to finish any interval 
(and hence gain no value) while $OPT$ will gain some.
Eventually, $OPT$ will accumulate at least $2-\delta$ times of
what $ALG$ can gain no matter what $ALG$ does in the last step.

\subsubsection{Step 1}
  \label{subsec:step1}

Let $c = 2 - \delta/2$.
The adversary releases $S_1 = SET(v_1,w_1,\epsilon_1)$
where $v_1$ is some positive real number at least one,
$w_1 = c(q/p)(4/\delta) v_1$ and $\epsilon_1 = \delta/8$.
Denote by $I_1$ and $J_1$, where $w(I_1) \leq w(J_1)$,
the intervals chosen by $ALG$.
We claim that

\begin{lemma} \label{lem:x1>v1}
Both algorithms $A$ and $B$ do not process the smallest-weight
interval in $S_1$, i.e.,
\begin{description}
\item[{\rm (i)}]
$w(I_1) > v_1$ and
\item[{\rm (ii)}]
$w(J_1) > v_1$.
\end{description}
Hence both $I_1^-$ and $J_1^-$ exist.
\end{lemma}

\begin{proof}
We first prove part (ii).
By Lemma \ref{lem:AHeavy}, we assume that $I_1$ is processed by $A$
and $J_1$ is processed by $B$.
So, the expected gain by $ALG$ is $p w(I_1) + q w(J_1)$ $\leq w(J_1)$.
Then we deduce that $w(J_1) > w_1/c$ or else the adversary stops,
$OPT$ schedules the heaviest interval in $S_1$ (of weight $w_1$)
so that it gains at least $c > 2-\delta$ times the expected gain
by $ALG$. 
Since $v_1 = (p/q)(\delta/4)(w_1/c)$ $< w_1/c$,
we have $w(J_1) > v_1$.

To prove part (i), we assume to the contrary that $w(I_1)=v_1$.
Then an interval $J'_1$ with the same weight as $J_1$
is released between $d(J_1^-)$ and $d(J_1)$.
See Figure \ref{fig:step1}(a).
$OPT$ can gain $w(J^-_1) + w(J'_1) \geq 2w(J_1) - \epsilon_1$
by executing $J^-_1$ in $S_1$ and then $J'_1$.
Upon finishing $I_1$, algorithm $A$ can go on to finish $J'_1$.
The expected gain of $ALG$ is at most
$p (w(I_1) + w(J'_1)) + q w(J_1)$
$= pw(I_1) + w(J_1)$.
Note that 
$v_1 = (p/q)(\delta/4)(w_1/c) \leq (\delta/4) w(J_1)$
and $p w(I_1) = p v_1$ $< (\delta/4) w(J_1)$.
Thus $ALG$'s gain is at most
$p w(I_1) + w(J_1)$ $\leq (1 + \delta/4) w(J_1)$. 
So $OPT$'s gain is more than $2-\delta$ times that of
$ALG$'s.
\end{proof}

\begin{figure}
\centerline{ \epsfysize=1.1in \epsffile{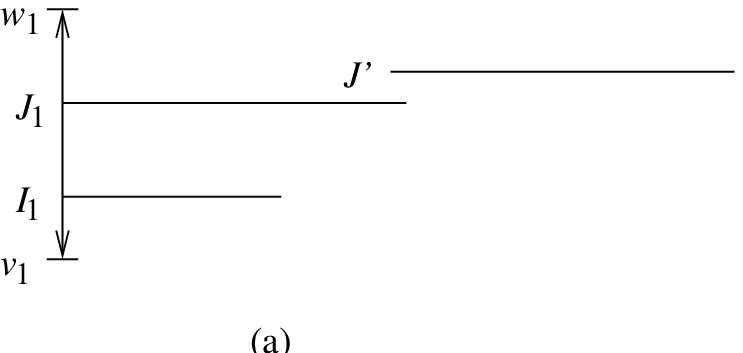} 
             \hspace{0.8in}
             \epsfysize=1.40in \epsffile{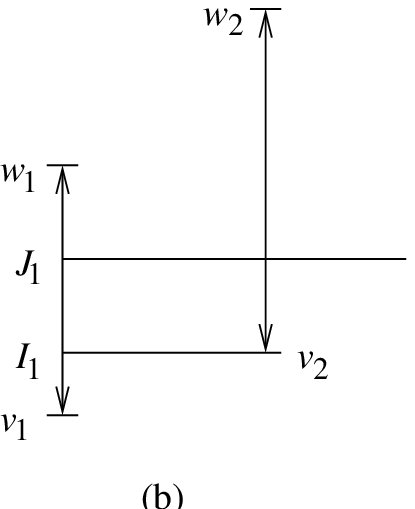} }
\caption{Step 1. (left) Lemma \ref{lem:x1>v1}, 
  (right) Lemma \ref{lem:mustmoveon} }
\label{fig:step1}
\end{figure}

\begin{lemma} \label{lem:case1}
$w(J_1) < 2 w(I_1)$.
\end{lemma}

We defer this to Section \ref{subsec:ylarge},
where we prove that if $w(J_1) \geq 2 w(I_1)$ then the adversary
can force the competitive ratio to be at least $2-\delta$.

The adversary then releases a new set of intervals 
$S_2 = SET(v_2,w_2,\epsilon_2)$
such that all these intervals arrive between
$d(I^-_1)$ and $d(I_1)$,
where $v_2 = w(I_1)$,
$w_2 = \max\{ c(p w(I_1) + q w(J_1)) - w(I_1), v_2\}$
and $\epsilon_2 = \epsilon_1/2$.
See Figure \ref{fig:step1}(b).

\begin{lemma} \label{lem:mustmoveon}
Upon the release of $S_2$, 
both $A$ and $B$ must abort their current intervals in $S_1$
and start some intervals $I_2$ and $J_2$ respectively in $S_2$.
Moreover, $v_2 < w(I_2) < w(J_2)$.
\end{lemma}
\begin{proof}
If $ALG$ ignores $S_2$ and continues with both $I_1$ and $J_1$,
then the expected gain of $ALG$ is $pw(I_1) + qw(J_1)$
while $OPT$ can complete $I^-_1$ and the last interval in $S_2$,
gaining $w(I^-_1) + w_2$ $\geq c(pw(I_1) + qw(J_1)) - \epsilon_1$ 
$= (2-\frac{\delta}{2}) (p w(I_1) + q w(J_1)) - \frac{\delta}{8}$
$\geq (2-\delta) (p w(I_1) + q w(J_1))$.

Suppose algorithm $B$ aborts $J_1$ and starts some $J_2$ in
$S_2$ while algorithm $A$ continues to process $I_1$.
Then by Lemma \ref{lem:ABSeparate}, $ALG$ loses.
Suppose algorithm $B$ continues with $J_1$ but 
$A$ aborts $I_1$ to start some $I_2$ in $S_2$.
By Lemma \ref{lem:AHeavy}, it must be the case that
$w(I_2) <w(J_1)$.
But then by Lemma \ref{lem:YEarly}, $ALG$ loses too.

Based on the above discussion, 
the only remaining sensible response for $ALG$
is to abort both $I_1$ and $J_1$
and start some $I_2$ and $J_2$ in $S_2$.
By Lemma \ref{lem:AHeavy},
we can further assume that $w(I_2) < w(J_2)$.
Moreover, we claim that $w(I_2) > v_2$.
Otherwise, $ALG$ is effectively aborting only $J_1$ but not $I_1$.
Then the construction of inputs given in Lemma \ref{lem:ABSeparate}
can be used to defeat $ALG$.
\end{proof}

This finishes our discussion on Step 1 and we now proceed to Step 2.

\subsubsection{Step $i$}

In general, at the beginning of Step $i \geq 2$,
we have the following situation:
$OPT$ has gained $w(I_1^-) + w(I_2^-) + \cdots w(I_{i-1}^-)$
while $ALG$ has not gained anything yet.
Moreover, $A$ and $B$ of $ALG$ 
are respectively executing $I_{i}$ and $J_{i}$
in $S_i = SET(v_i,w_i,\epsilon_i)$ 
where $v_i = w(I_{i-1})$, 
$w_i = \max\{ c(p w(I_{i-1}) + q w(J_{i-1})) - \sum_{j=1}^{i-1} w(I_j), v_i \}$,
$\epsilon_i = \epsilon_1/2^{i-1}$
and $v_i < w(I_i) < w(J_i)$.
We go through a similar analysis to that in Step 1.

First, as in Lemma \ref{lem:case1}, we have $w(J_i) < 2 w(I_i)$
(the case $w(J_i) \ge 2 w(I_i)$ is handled in the next subsection). 
Next, the adversary releases $S_{i+1} = SET(v_{i+1}, w_{i+1}, \epsilon_{i+1})$ 
in the period between $d(I^-_i)$ and $d(I_i)$
where $v_{i+1} = w(I_i)$,
$w_{i+1} = \max\{ c(p w(I_i) + q w(J_i)) - \sum_{j=1}^i w(I_j), v_{i+1}\}$
and $\epsilon_{i+1} = \epsilon_1/2^i$.
See Figure \ref{fig:stepi}.

\begin{figure}
\centerline{ \epsfysize=1.3in \epsffile{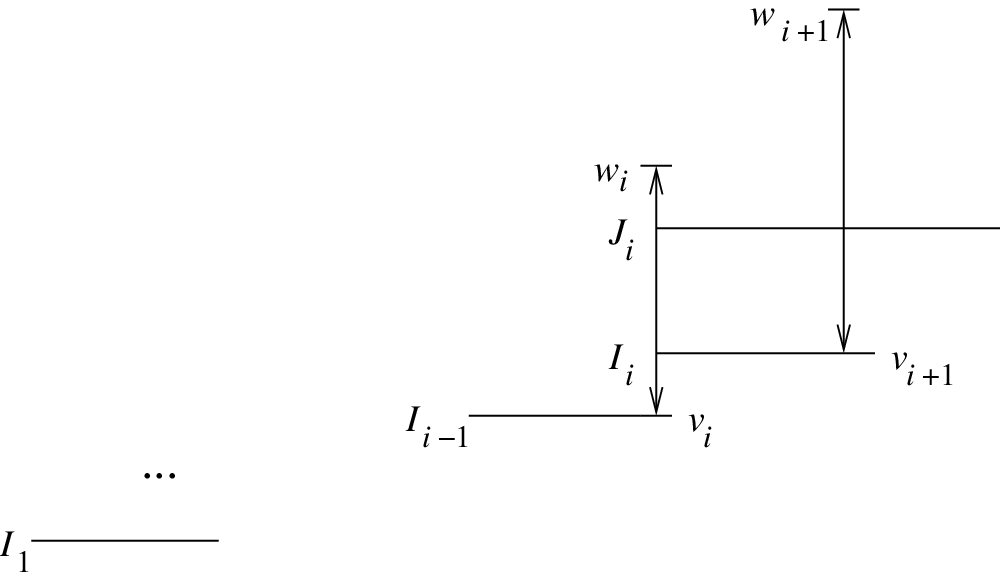} }
\caption{Step $i$}
\label{fig:stepi}
\end{figure}

Similar to Lemma \ref{lem:mustmoveon}, we can prove that

\begin{lemma} 
Upon the release of $S_{i+1}$,
both $A$ and $B$ must abort their current intervals in $S_i$
and start some intervals $I_{i+1}$ and $J_{i+1}$ respectively in $S_{i+1}$.
Moreover, $v_{i+1} < w(I_{i+1}) < w(J_{i+1})$.
\end{lemma}

\begin{proof}
$ALG$ cannot continue with both $I_i$ and $J_i$.
Otherwise, $OPT$ schedules (after finishing $I_1^-, \ldots, I_{i-1}^-$)
$I_i^-$ and then the last interval of $S_{i+1}$,
thus gaining at least
\begin{eqnarray*}
\lefteqn{ \sum_{j=1}^i w(I_j^-)      
            + c( p w(I_i) + q w(J_i)) - \sum_{j=1}^i w(I_j) } \\
  & \geq & c (p w(I_i) + q w(J_i)) - \sum_{j=1}^i \epsilon_j  \\
  &  >   & (2-\frac{\delta}{2}) (p w(I_i) + q w(J_i)) - 2\epsilon_1   \\
  & \geq & (2-\delta) (p w(I_i) + q w(J_i)).
\end{eqnarray*}

Suppose $B$ aborts $J_i$ in order to start some $J_{i+1}$ in $S_{i+1}$ 
while $A$ continues with $I_i$.
Then by Lemma \ref{lem:ABSeparate}, $ALG$ loses.
Suppose $B$ continues with $J_i$ while
$A$ aborts $I_i$ to start $I_{i+1}$.
By Lemma~\ref{lem:AHeavy}, we have that $w(I_{i+1}) < w(J_i)$.
Then by Lemma \ref{lem:YEarly}, $ALG$ loses too.

Based on the above reasoning, we conclude that $ALG$ has to abort
both $I_i$ and $J_i$ and start some $I_{i+1}$ and $J_{i+1}$ in $S_{i+1}$.
By Lemma \ref{lem:AHeavy}, we can assume that $w(I_{i+1}) < w(J_{i+1})$.
We can also argue that $w(I_{i+1}) > v_{i+1}$ in the same way as proving
$w(I_2) > v_2$ in Step 1.
\end{proof}

We now proceed to Step $i+1$.
Note that $OPT$ has already gained
$w(I_1^-) + \cdots + w(I_i^-)$ while $ALG$ still has not gained anything.
We will make use of Lemma 4.3 of Woeginger \cite{Woeg94}:

\begin{lemma} \label{lem:woeg}
(Woeginger \cite{Woeg94}).
For $2<d<4$, any strictly increasing sequence of positive
numbers $\langle a_1, a_2, \ldots \rangle$ fulfilling
the inequality
  \begin{eqnarray*}
  a_{i+1} \leq d a_{i} - \sum_{j=1}^i a_j
  \end{eqnarray*}
for every $i \geq 1$ must be finite.
\end{lemma} 

Consider the sequence $\langle w(I_1), w(I_2), \ldots \rangle$.
It is strictly increasing since $w(I_{i+1}) > v_{i+1} = w(I_i)$ for all $i$.
Moreover, recall that $w_{i+1}$ is set to be the maximum of either
$c(p w(I_i) + q w(J_i)) - \sum_{j=1}^i w(I_j)$ or $v_{i+1}$.
If $c(p w(I_i) + q w(J_i)) - \sum_{j=1}^i w(I_j) > v_{i+1}$
for all $i \geq 1$,
then we have $w(I_{i+1}) \leq w_{i+1}$
$= \max \{ c(p w(I_i) + q w(J_i)) - \sum_{j=1}^i w(I_j), v_{i+1} \}$
$= c(p w(I_i) + q w(J_i)) - \sum_{j=1}^i w(I_j)$
$\leq c w(J_i) - \sum_{j=1}^i w(I_j)$
$< (4-\delta) w(I_i) - \sum_{j=1}^i w(I_j)$
(since $w(J_i) < 2w(I_i)$) for all $i$.
The existence of such an infinite sequence contradicts Lemma \ref{lem:woeg}.
So, eventually, there is a finite $k$ such that 
$c(p w(I_k) + q w(J_k)) - \sum_{j=1}^k w(I_j) \leq v_{k+1}$
and hence we set $w_{k+1} = v_{k+1}$ and the next (and final)
set $S_{k+1}$ consists of a single interval of weight $w_{k+1}$
($=w(I_k)$).

In that situation, it makes no difference whether $A$ or $B$ of $ALG$
aborts $I_k$ or $J_k$ to start the interval in $S_{k+1}$
since it has weight equal to $w(I_k)$.
Its expected gain is still at most $p w(I_k) + q w(J_k)$.
On the other hand, $OPT$ schedules the interval of $S_{k+1}$
and gains in total $w(I_1^-) + \cdots + w(I_{k}^-) + w_{k+1}$
$\geq c(p w(I_k) + q w(J_k)) - 2\epsilon_1$ which is
at least $2-\delta$ times of $ALG$'s gain.

\subsection{The Case of $w(J_i) \ge 2 w(I_i)$} 
  \label{subsec:ylarge}

We now consider the case where in some Step $i \ge 1$,
$w(J_i) \ge 2w(I_i)$.
We will show how the adversary forces $ALG$ to lose the game,
i.e., $OPT$ will gain at least $2-\delta$ of what $ALG$ can
on $S_i$ and a set of subsequently arrived intervals.
For simplicity, we drop the subscript $i$ in $I_i$, $J_i$, and
$\epsilon_i$ in the following discussion.

Intuitively, when $w(J)$ is relatively large compared with $w(I)$,
we can afford to let algorithm $A$ finish the interval $I$
and gain $pw(I)$, which is relatively small.
Therefore, a set $S_u = SET(0,u w(J),\epsilon)$ is released 
between $d(J^-)$ and $d(J)$,
where $u \geq 1$ is some parameter to be determined
as a function of $p$.
This allows algorithm $A$ to finish $I$ and then start
some job in this new set $S_u$.
On the other hand, algorithm $B$ has to decide whether to
abort or continue with the current interval $J$.
We will show that there is a choice of $u$ such that no matter
what $B$ does, $OPT$ can gain at least $2-\delta$ of what $ALG$
gains.

\paragraph{Case 1: Algorithm $B$ continues with $J$.}

Then $ALG$ gains at most $p(w(I) + uw(J)) + qw(J)$
while $OPT$ can gain $w(J^-) + uw(J)$.
Therefore, the ratio of the gain by $OPT$ to that of $ALG$
on $S$ and $S_u$ is at least
\begin{eqnarray*}
\frac{ w(J^-) + uw(J) } { p(w(I)+uw(J)) + qw(J) } 
  & \geq & \frac{(1+u)w(J) - \epsilon} { p (w(J)/2 + uw(J)) + (1-p)w(J) } \\
  & \geq & \frac{1+u - \epsilon/w(J)} { p (1/2 + u) + 1-p }
\end{eqnarray*}
where the first inequality makes use of the condition that
$w(I) \leq w(J)/2$.
This ratio is at least $2-\delta$ provided
\begin{eqnarray*}
1 + u - \frac{\epsilon}{w(J)}
   & \geq & (2-\delta)(1 - \frac{p}{2} + pu) \\
   &  =   & (2 - \delta - p + \frac{p\delta}{2})
               + (2p - p\delta) u.
\end{eqnarray*}
After simplifying using $\epsilon/w(J) \leq \epsilon < \delta$,
this condition is satisfied by having
\begin{eqnarray}
u & \geq & \frac{1 - p + p\delta/2}
                {1 - 2p + p\delta}.
\end{eqnarray}

\paragraph{Case 2: Algorithm $B$ aborts $J$.}

Then algorithms $A$ and $B$ can start some intervals
$I'$ and $J'$ in $S_u = SET(0, uw(J), \epsilon)$.
By Lemma \ref{lem:AHeavy},
we can assume that $w(I') \leq w(J')$.
Then another interval $J''$ with the same weight as $J'$ 
is released between $d((I')^-)$ and $d(I')$.

If $ALG$ aborts $I'$ to start $J''$,
by Lemma \ref{lem:YEarly}, $OPT$ gains at least
$2-\epsilon$ of $ALG$'s gain on $S_u$ and a set of subsequently
arrived intervals.
Also, on the set of intervals $S = SET(v,w,\epsilon)$,
$ALG$ gains $pw(I) \leq w(I)/2$ while $OPT$ gains at least 
$w(J) - \epsilon$ $\geq w(I) -\epsilon$.
So $ALG$ loses.

On the other hand, if $ALG$ does not abort $I'$,
then its expected gain is at most $p(w(I)+w(I')) + qw(J')$
$\leq p(w(J)/2 + w(I')) + (1-p)w(J')$.
$OPT$ will complete $J^-$ in $S$,
$(I')^-$ in $S_u$ and then $J''$.
(Note that if $I'$ is the first interval with weight 0
in $S_u$, then we set $(I')^-$ to be of weight 0 too.)
The ratio of the gain by $OPT$ to that of $ALG$ 
on $S$, $S_u$ and $J''$ is at least
\begin{eqnarray*}
\frac{(w(J)-\epsilon) + (w(I')-\epsilon) + w(J'')}
     {p(w(J)/2 + w(I')) + (1-p)w(J')} 
    =   \frac{w(I') + w(J) + w(J') - 2\epsilon}
               {pw(I') + pw(J)/2 + (1-p) w(J')}.
\end{eqnarray*}
Since $1/p \ge 2 > 2-\delta$, it suffices to show that
\begin{eqnarray*}
f & = & \frac{w(J) + w(J') -2\epsilon}
             {pw(J)/2 + (1-p) w(J')} 
\end{eqnarray*}
is at least $2-\delta$.
Furthermore, $1/(p/2) \geq 2$ and
$\frac{1- 2\epsilon/w(J')}{1-p}$ $< \frac{1}{1-p} \leq 2$.
Therefore, the fraction $f$, as a function of $w(J')$,
is minimized when $w(J')$ is maximized.
That means,
\begin{eqnarray*}
f & \geq & \frac{w(J) + uw(J) - 2\epsilon}
                {pw(J)/2 + (1-p) uw(J)}   \\
  & \geq & \frac{1 + u -2\epsilon}
                {p/2 - pu + u}.
\end{eqnarray*}
Observe that $\frac{p}{2} - pu + u$
$\geq \frac{p}{2} - p + 1$
$\geq 1 - \frac{p}{2}$ $> \frac{1}{2}$
using $u \geq 1$ and $p \leq 1/2$.
Thus $\frac{2\epsilon}{p/2 - pu + u} \leq \frac{\delta}{2}$
because $\epsilon \leq \delta/8$.
Hence to show that $f \geq 2-\delta$, it suffices to
show that $\frac{1+u}{p/2 - pu + u} \geq 2-\frac{\delta}{2}$.
This condition is satisfied provided
\begin{eqnarray*}
1 + u
  & \geq & (2-\frac{\delta}{2}) (\frac{p}{2} - pu + u) \\
  &  =   & (p - \frac{p\delta}{4}) 
              + (2 - 2p + \frac{p\delta}{2} - \frac{\delta}{2}) u.
\end{eqnarray*}
or
\begin{eqnarray*}
(1 - 2p + \frac{p\delta}{2} - \frac{\delta}{2}) u 
   \leq 1 - p + \frac{p\delta}{4}.
\end{eqnarray*}
This is satisfied if
\begin{eqnarray*}
(1 - 2p + \frac{p\delta}{2}) u \leq 1 - p + \frac{p\delta}{4},
\end{eqnarray*}
i.e.,
\begin{eqnarray}
u & \leq & \frac{1 - p + p\delta/4}
                {1 - 2p + p\delta/2}.
\end{eqnarray}
By setting
\begin{eqnarray*}
u & = & \frac{1 -p + p\delta/4}
             {1 -2p + p\delta/2},
\end{eqnarray*}
both ineq. (1) and ineq. (2) are satsified.
This completes the proof for the case of $w(J_i) \geq 2 w(I_i)$.

\section{Conclusion}

In this paper, we designed 2-competitive barely random algorithms
for various versions of preemptive scheduling of intervals.
They are surprisingly simple and yet improved upon previous best
results.
Based on the same approach, we designed a 3-competitive algorithm
for the preemptive (with restart) scheduling of equal length jobs.
This is the first randomized algorithm for this problem.
Finally, we gave a 2 lower bound for barely random algorithms
that choose between two deterministic algorithms, possibly
with unequal probability.

An obvious open problem is to close the gap between the upper and
lower bounds for randomized preemptive scheduling of intervals
in the various cases.
We conjecture that the true competitive ratio is 2.
Also, it is interesting to prove a randomized lower bound for
the related problem of job scheduling with restart.

\bibliographystyle{alpha}


\end{document}